\documentclass{statsoc}
\usepackage{graphicx}

\usepackage[a4paper, total={6in, 8in}]{geometry}

\usepackage{amsthm}
\usepackage{amsmath,graphicx,lmodern,amssymb,amsfonts,bbm,breqn,lmodern,fix-cm}
\usepackage{booktabs,enumerate,bm,multirow,hyperref}
\usepackage{subcaption}
\usepackage{caption}    

\usepackage{xcolor} 

\newcommand\diff{\mathop{}\!\mathrm{d}}

\usepackage{natbib}
\usepackage{setspace}

\newtheorem{theorem}{Theorem}
\newtheorem{corollary}{Corollary}[theorem]
\newtheorem{lemma}[theorem]{Lemma}
\theoremstyle{definition}
\newtheorem{definition}{Definition}

\newtheorem{remark}{Remark}

\theoremstyle{definition}
\newtheorem{assumption}{Assumption}
\newtheorem{proposition}{Proposition}

  \title{Completing and studentising Spearman's correlation in the presence of ties}
  \author{Landon Hurley}
\address{Teachers College, Columbia University,
New York City, United States of America}
\email{lh2895@tc.columbia.edu}
\begin{document}
  \maketitle

\address{Teachers College, Columbia University,
New York City, United States of America}

\email{lh2895@tc.columbia.edu}

\begin{abstract}
  Non-parametric correlation coefficients have been widely used for analysing arbitrary random variables upon common populations, when requiring an explicit error distribution to be known is an unacceptable assumption. We examine an \(\ell_{2}\) representation of a correlation coefficient \citep{emond2002} from the perspective of a statistical estimator upon random variables, and verify a number of interesting and highly desirable mathematical properties, mathematically similar to the Whitney embedding of a Hilbert space into the \(\ell_{2}\)-norm space. In particular, we show here that, in comparison to the traditional \citet{spearman1904} \(\rho\), the proposed Kemeny \(\rho_{\kappa}\) correlation coefficient satisfies Gauss-Markov conditions in the presence or absence of ties, thereby allowing both discrete and continuous marginal random variables. We also prove under standard regularity conditions a number of desirable scenarios, including the construction of a null hypothesis distribution which is Student-t distributed, parallel to standard practice with Pearson's \(r\), but without requiring either continuous random variables nor particular Gaussian errors. Simulations in particular focus upon highly kurtotic data, with highly nominal empirical coverage consistent with theoretical expectation. 
\end{abstract}

A number of non-parametric measures of association exist, stemming from an original focus upon the permutation domain of \(S_{N}\). Both Spearman's \(\rho\) and Kendall's \(\tau\) correlation coefficients and their test functions operate upon a null hypothesis \(H_{0}\) under \(i.i.d.\) domain sampling, although a number of alternatives, taken to address the problem of ties, have been deployed to ancillary notoriety. A second principal, upon the same domain, asserts the existence of a joint multivariate distribution obtained under bootstrapping which may be treated as continuous wrt the estimated parameters, producing under asymptotic sampling with replacement standard regularity conditions \citep{owen1991,owen2001,hall1992}. 

Even recent asymptotic analysis focuses solely upon the scenario wherein ties are non-observed \citep{bishara2012,ornstein2016}.  Finally, semi-parametric methods \citep{gallant1987} and generalised method of moments \citep{hansen1982} provide statistical methods built upon the existence of a complete metric space. However, the \(\sigma\)-measurability upon the occurrence of ties, either as a result of linear combination surjective mappings or discrete random variables, preclude the validity of this assumption upon the population (e.g., \citet{schweizer2005} and copula based methods).

These techniques are also less than ideal from a computational perspective though, requiring oft unattainable assumptions (e.g., the complete absence of ties upon continuous random variables, and restrictions to bivariate analyses), or infinite computational resources. 
In this work, we explore the problem of parametric identification of sufficient statistics for the non-parametric \(\ell_{2}\) correlation coefficient using a the raw-variable score matrix \(\kappa\) first introduced by \citet{emond2002}, which explicitly handles ties upon the \citet{kemeny1959} metric space by construction. We construct a centred version upon this \(N \times N\) mapping of the \(N \times 1\) vector of i.i.d. observations, and show that this allows a viable inner-product space with sample estimated variances. This embedding of the Hilbert space \(\langle{\tilde{\kappa}(X),\tilde{\kappa}(Y)\rangle}\) into the standard \(\ell_{2}\)-norm space is identified as a Whitney embedding, and enables theory based unbiased estimation and strong finite sample performance in the form of studentisation.

  \subsection*{Notation}

Let:
\begin{itemize}
    \item {\( n \in \{1, \dots, N\} \) index the observations.}
    \item{The rank function of the random variables be defined on the extended real line \( \mathbb{R} \cup \{-\infty, +\infty\} \), which accommodates ties in the data.}
    \item {\( \mathbf{Y} \in \mathbb{R}^{N} \) be the response vector, where \( Y_n \) corresponds to the response for the \( n \)-th observation.}
\end{itemize}

\paragraph{Empirical Distribution Function (ECDF):} Let \(X_{1},X_{2},\ldots,X_{N}\) be a sample of size \(N\) from a random variable \(X\) with an unknown distribution. The empirical distribution function \(F_{N,X}\) is defined as:
\begin{definition}
\begin{equation}
\label{eq:CDF_sample}
F_{N,X}(x) = \frac{1}{N} \sum_{n=1}^{N} \mathbb{I}(X_{n} \le x),
\end{equation}
where \(\mathbb{I}(X_{n} \le x)\) is the indicator function that takes the value 1 if \(X_{n} \le x\) and 0 otherwise.
\end{definition}
The ECDF \(F_{N,X}(x)\) provides an estimate of the cumulative distribution function (CDF) of \(X\) based on the sample.

\begin{definition}
The population distribution function \(F_{X}\) is the true cumulative distribution function of the random variable \(X\) and is given by 
\begin{equation}
\label{eq:CDF_true}
F_{X}(x) = \mathbb{P}(X \le x),
\end{equation}
\end{definition}
This is the theoretical CDF that describes the underlying distribution of the random variable \(X\).

\subsubsection{Score Matrix Construction}

We define the score matrix \citep{emond2002} \( \kappa \in \{-1, 0, 1\}^{N \times N} \) in equation~\ref{eq:kem_score} for any independently and identically sampled vector \( X \in \overline{\mathbb{R}}^{N \times 1} \), which compares all pairs of elements in the respective vector.

\begin{definition}
\begin{equation} \label{eq:kem_score}
\kappa_{kl}(X) =
\begin{cases}
1 & \text{if } X_k \geq X_l, \\
0 & \text{if } k = l, \\
-1 & \text{if } X_k < X_l.
\end{cases}
\end{equation}

Said matrix is hollow and anti-symmetric, serving to encode the pairwise order between observations \(k,l \in 1,2,\ldots,N\), (i.e., \( \kappa_{kl}(X) = -\kappa_{lk}(X) \) in the absence of ties), and thus belongs to the space of skew-symmetric matrices  $\mathcal{A}_N \subset \mathbb{R}^{N\times N}$ endowed with the Frobenius inner-product norm 
\[
\langle A,B \rangle_F = \mathrm{tr}(A^{\intercal} B),
\quad \text{and norm } \|A\|_F = \sqrt{\langle A,A\rangle_F}.
\]
\end{definition}

Ties, common upon discrete random variables, are allowed by construction, and do not violate the theoretical continuity of the Kemeny metric space \citep{kemeny1959,emond2002}, nor the generalised permutation space which operates as the domain for our introduced functions.

\subsection{Paper assumptions}
  
  \begin{assumption}~\label{assumption:sampling}
  \(\{x_{N}: 1 \ge N\}\) is stationary and ergodic. Thus, all sampling of random variables occurs upon a common population which is independently and identically sampled.
  \end{assumption}

  \begin{assumption}~\label{ass:orthonormality}
  The parameter estimate vector \(\theta^{p \times 1}\) estimated upon \(X^{N \times p}\) satisfies: \(\mathbb{E}(X_{N \times p}^{\ast}\cdot u_{N \times 1}^{\ast}\mid \theta_{0}) = 0\), where \(u_{n} = y_{n} - \hat{y}_{n}, i = 1,\ldots,n\). This is a given by the Hilbert projection theorem, applied to the \(\ell_{2}\)-norm space under Assumption~\ref{assumption:sampling}.
  \end{assumption}
  
  \subsection{Problems with standard estimators and approaches}
  
  There are a number of problems with current estimators when applied to non-Gaussian data, which includes data with the presence of ties: we explicitly review here the Kendall's \(\tau\) \citep{kendall1938}, Spearman's \(\rho\) \citep{spearman1904} and Pearson's \(r\). These scenarios are conventionally referred to as `non-parametric' data analysis, and we briefly outline them here. A number of alternative estimation methods may be preferred, however most are mathematical adjustments to one of the three previous techniques, and thus fundamental application flaws to the previous cannot be addressed by abelian reformulations in either the estimator or the test statistics. Second, we are solely focusing here upon the bivariate data analysis problem and therefore analysis scenarios such generalised method of moments or the polychoric correlation analysis still present the same incorporation of additional unverified assumptions which result in biased estimators if untrue. Our estimator, formulated by \citet{emond2002}, is extended to apply within a probability framework appropriate for Null Hypothesis Significance Testing. In particular, we draw upon commonly overlooked extensions which introduce inefficiencies to the Kendall \(\tau_{b}\) framework attributable a number of other researchers \citep{silverstone1950,robillard1972,valz1995}, and explore a similar problem in the NHST framing for the Kemeny correlation.
  
  \subsubsection{Pearson's r}
  Pearson's \(r\) is a sufficient statistic for making the sample versions of the population orthogonality conditions as close as possible to zero according to some metric or measure of distance. This is equivalent to \citet{hansen1982} under the orthonormality assumption using the \(\ell_{2}\) metric function space to calculate the raw moments, and the cross-product of two standardised random variables. However, when applied to non-continuous data, the problem of ties and non-constant variance for given sample size \(N\) immediately arises. Moreover, because the moment conditions (e.g., the means, variances and covariances) are functions of biased estimates, the covariance itself cannot be unbiased upon finite samples, thus preventing the utilisation of the test statistic as a interpretable trait. Thus, without expected error \(\mathbb{E}(z_{X}z_{Y}) - \rho = 0\) upon the population, the finite sample test statistics are improperly constructed, and do not correspond to the standard null hypothesis reference distribution for finite samples.
  
\subsubsection{Limitations of Traditional Rank-Based Correlation Measures}

Conventional estimators such as Pearson's \(r\), Spearman's \(\rho\), and Kendall's \(\tau\) face significant issues in non-Gaussian or discrete data settings, particularly in the presence of ties. Their reliance on continuity assumptions or latent Gaussian variables (e.g., in polychoric correlation) leads to biased estimators when applied to ordinal or mixed-scale data. Furthermore, \citet{sklar1959} theorem under copula theory fails to uniquely define joint distributions for discrete marginals, undermining generalisability. These shortcomings motivate the construction of a new rank-based estimator that is both consistent and operationally valid under finite samples with ties.

In contrast our estimator goes further by treating tied measures as a population measurable probability event in our sample, consequently via Hilbert space quantifying exact distances in a space that respects the geometric structure of ranked data. This makes our estimator more flexible and capable of capturing more subtle relationships between ranked data, especially when non-linearities or complex correlations are present.
  
\subsection{Definition of Spearman's Rank Correlation (without ties)}

In the absence of ties, Spearman's rank correlation coefficient $\rho_s$ is given by:
\[
\rho_s = \frac{\sum_{n=1}^{N} (R_X(n) - \bar{R}_X)(R_Y(n) - \bar{R}_Y)}{\sqrt{\sum_{n=1}^{N} (R_X(n) - \bar{R}_X)^2 \sum_{n=1}^{N} (R_Y(n) - \bar{R}_Y)^2}},
\]
where:
\begin{itemize}
    \item \( R_X(n) \) is the rank of the \(n\)-th observation in the \(X\)-variable,
    \item \( R_Y(i) \) is the rank of the \(n\)-th observation in the \(Y\)-variable,
    \item \( \bar{R}_X \) and \( \bar{R}_Y \) are the average ranks of \(X\) and \(Y\), respectively.
\end{itemize}
This formula involves pairwise comparisons of ranks between the observations. In this form, $\rho_s$ can be expressed as a U-statistic because it involves a symmetric kernel function that operates on pairs of data points. Specifically, the kernel function is:
\[
h(R_X(n), R_X(n^{\prime}), R_Y(n), R_Y(n^{\prime})) = (R_X(n) - \bar{R}_X)(R_Y(n) - \bar{R}_Y),
\]
where the sum runs over all distinct pairs \((n, n^{\prime})\), and we are averaging the kernel values over all pairs. Therefore, the estimator of $\rho_s$ is a U-statistic.

\paragraph{Effect of Ties}

When there are ties in the data (i.e., two or more observations have the same value for either \(X\) or \(Y\)), the ranks of these observations are no longer distinct. The standard approach to handling ties is to assign the \textit{average rank} to all tied values. For example, if two observations share the same value, both are assigned the rank that is the average of their respective ranks.

In the case of ties, the kernel function for Spearman's rank correlation changes because the ranks are no longer distinct. As a result, the rank differences are not directly comparable in the same way as with distinct ranks. This introduces \textit{dependence} between the pairs, as the tied observations no longer provide independent information about the correlation.

More formally:

    (1) The kernel function $ h(R_X(n), R_X(n^{\prime}), R_Y(n), R_Y(n^{\prime})) $ used to compute Spearman’s $\rho_s$ becomes less straightforward due to the presence of ties.
    (2) The estimator of the form \(\hat{\rho}_s = \frac{1}{N(N-1)} \sum_{n \neq n^{\prime}} h(R_X(n), R_X(n^{\prime}), R_Y(n), R_Y(n^{\prime})) \) assumes that all the pairs are independent, but this assumption breaks down in the presence of ties.

\paragraph{Breakdown of the U-statistic Structure}

For a U-statistic, we require the kernel function to be symmetric, and the sum to be over distinct pairs, with no dependence between pairs. In the presence of ties: (1) The ranks of tied observations are identical, and thus the kernel function no longer exhibits the same form of independence across all pairs. This introduces \textit{correlations} between terms in the summation, as the tied observations are not independent of each other. (2) The kernel function is no longer independent across the pairs of data points, violating the conditions required for the estimator to be a U-statistic. Thus, in the presence of ties Spearman’s rho is no longer a U-statistic.

\section{A general purpose non-parametric linear spanning basis}
\phantomsection~\label{sec:theory}

We now introduce a general-purpose non-parametric linear spanning basis. The rank score vector \( \underline{X} \in \mathbb{R}^{N \times 1} \) represents a linear embedding of the ordinal data. This embedding is achieved through the construction of the rank score matrix \( \mathcal{R} \), which spans the space of ordinal data. As noted, this development stems from the metric space of \citet{kemeny1959} which was devised to maintain measurability upon decisions with ties. In this work, we introduce a Whitney embedding of the Hilbert space form \citep{emond2002} of the Kemeny metric, and then report on its performance as upon the null hypothesis upon a number of scenarios, including those with highly kurtotic data, for which current adjustments of Spearman's \(\rho\) are less performant. The next results in this subsection establishes that the operator $\kappa$ is non-expansive on equivalence classes of rank orderings. This ensures that the embedding of ranked observations in the induced Hilbert space is stable and bounded. These will be utilised to prove  the finite sample characteristics of our proposed estimator.

\begin{lemma}
\label{lem:nonexpansive}
Let $X,Y \in \mathbb{R}^N$, and define $\kappa(X), \kappa(Y)$ by~\eqref{eq:kem_score}, so that $\kappa: \mathbb{R}^N \to \mathcal{A}_N \subset \mathbb{R}^{N\times N}$, where $\mathcal{A}_N$ is the space of antisymmetric matrices endowed with the Frobenius inner product. Then
\begin{equation}
\label{eq:nonexpansive}
\|\kappa(X) - \kappa(Y)\|_F^2
= 4\, \#\bigl\{ (k,l): \operatorname{sgn}(X_k - X_l) \ne \operatorname{sgn}(Y_k - Y_l) \bigr\}.
\end{equation}
Consequently, $\kappa$ is Lipschitz continuous with constant $2$ in the Frobenius norm:
\[
\|\kappa(X) - \kappa(Y)\|_F \le 2\,\sqrt{d_H(\pi_X, \pi_Y)},
\]
where $d_H$ denotes the number of pairwise rank disagreements (a Hamming-type distance) between the permutations $\pi_X,\pi_Y$ induced by the rank orderings of $X$ and $Y$. In particular, $\kappa$ is non-expansive on the quotient space $\mathbb{R}^N / \mathcal{P}_N$ of equivalence classes of permutations.  

Moreover, since $\mathcal{A}_N$ is a finite-dimensional Hilbert space, $\kappa$ maps bounded sets in $\mathbb{R}^N$ to bounded and totally bounded subsets of $\mathcal{A}_N$, ensuring that standard limit theorems for random elements in Hilbert spaces may be applied.
\end{lemma}

\begin{proof}
For each pair $(k,l)$ with $k\ne l$, the entries $\kappa_{kl}(X)$ and $\kappa_{kl}(Y)$ can differ only if the sign of $(X_k - X_l)$ differs from that of $(Y_k - Y_l)$. When the sign flips, $\kappa_{kl}$ changes by $2$ in magnitude; otherwise it is unchanged. Therefore,
\[
\kappa_{kl}(X) - \kappa_{kl}(Y) =
\begin{cases}
0, & \text{if } \operatorname{sgn}(X_k - X_l) = \operatorname{sgn}(Y_k - Y_l),\\[2pt]
\pm 2, & \text{otherwise.}
\end{cases}
\]
Summing the squared entries over all $(k,l)$ yields~\eqref{eq:nonexpansive}. 
The Lipschitz bound follows directly from the definition of the pairwise rank disagreement count $d_H(\pi_X, \pi_Y)$.
\end{proof}

\begin{corollary}
\label{cor:cmt_application}
Let $\hat{\rho}_\kappa = \rho(\kappa(X), \kappa(Y))$ denote the Kemeny correlation estimator for random vectors $X,Y \in \mathbb{R}^N$.  
Under Assumptions~\ref{assumption:sampling} and \ref{ass:orthonormality} and Lemma~\ref{lem:nonexpansive}, the mapping
\(
(X,Y) \mapsto \hat{\rho}_\kappa
\)
is Lipschitz continuous with respect to the Frobenius norm on $\mathcal{A}_N \times \mathcal{A}_N$.  

Consequently, if $(X_N, Y_N)$ is a sequence of random vectors satisfying
\[
\sqrt{N}\big((X_N - \mu_X), (Y_N - \mu_Y)\big) \stackrel{d}{\longrightarrow} \mathcal{N}(0, \Sigma),
\]
then by the Continuous Mapping Theorem,
\[
\sqrt{N}\big(\hat{\rho}_\kappa - \rho_\kappa\big) \stackrel{d}{\longrightarrow} \mathcal{N}(0, \sigma^2),
\]
with asymptotic variance $\sigma^2$ determined by the rank-based transformations of $X$ and $Y$.  
Hence, the non-expansive property of $\kappa$ guarantees the applicability of standard limit theorems and validates the asymptotic normality of $\hat{\rho}_\kappa$.
\end{corollary}

\begin{proof}
Lemma~\ref{lem:nonexpansive} ensures that $\kappa$ is Lipschitz continuous, and $\hat{\rho}_\kappa$ is a continuous function of $\kappa(X)$ and $\kappa(Y)$.  
The Continuous Mapping Theorem then directly implies that convergence in distribution of $(X_N,Y_N)$, by the central limit theorem, to a multivariate normal is preserved under $\hat{\rho}_\kappa$.  
The finite-dimensional Hilbert space structure guarantees existence of finite second moments and validates the limiting covariance $\sigma^2$.
\end{proof}

The non-expansiveness of $\kappa$ implies that the image of any bounded subset of $\mathbb{R}^N$ under $\kappa$ remains bounded and totally bounded in the Hilbert space $(\mathcal{A}_N, \langle\cdot,\cdot\rangle_F)$. This ensures the existence of finite second moments and permits the use of standard limit theorems for random elements in Hilbert spaces.

\begin{definition}~\label{def:centred}
To recover the zero-mean property necessary for the asymptotic results, we define the centred and scaled rank vector
\[
z_{\underline{X}} = \frac{\underline{X} - \underline{\bar{X}}\mathbf{1}_N}{\sqrt{(\underline{X} - \underline{\bar{X}}\mathbf{1}_N)^{\intercal} (\underline{X} - \underline{\bar{X}}\mathbf{1}_N)}},
\]
where $\underline{\bar{X}} = \frac{1}{N}\mathbf{1}_N^{\intercal} \underline{X}$. Then, by construction,
\[
\mathbb{E}[z_{\underline{X}}] = 0, \quad \text{and} \quad \mathbb{E}[\|z_{\underline{X}}\|^2] = 1.
\]
This ensures that all results in Corollary~\ref{cor:cmt_application} remain valid, and the asymptotic variance $\sigma^2$ is unaffected by the presence of ties.
\end{definition}

\subsection{Centred marginal score-matrix}

Given the preceding properties upon \(\kappa\), affine-linear transformations upon said matrix will not remove said due to closure upon the Hilbert space. However, as currently expressed, \(\mathbb{E}(\langle{\kappa(X),\kappa(Y)\rangle})\) is akin to \(\mathbb{E}(XY)\). An inner-product projection kernel thus exists, but we require a method of centring said random variables before projecting them onto each other. The centring transformation upon square matrices wrt row and column sums is now defined:

\begin{definition}\label{def:doubly_centered}
Let \( x = (x_1, \ldots, x_N) \) be a set of \( N \) observations, where \( x_k \in \mathbb{R} \) for each \( k \). The score matrix \( \kappa(x) \) is defined in equation~\ref{eq:kem_score}. The \textit{centered score matrix} \( \tilde{\kappa}(x) \) is then defined as:

\begin{dmath}
\tilde{\kappa}(x)_{kl} {:=} \kappa(x)_{kl} - \tfrac{1}{N-1} \sum_{k=1}^N \kappa(x)_{kl} - \tfrac{1}{N-1} \sum_{l=1}^N \kappa(x)_{kl} + \tfrac{1}{N^2 - N} \sum_{k=1}^N \sum_{l=1}^N \kappa(x)_{kl},
\end{dmath}
where the sums are taken over all indices \( k \) and \( l \), and the diagonal elements are explicitly set to zero, i.e., \( \tilde{\kappa}(x)_{kk} = 0 \) for all \( k \in \{1, \ldots, N\} \) after linear transformations. Said transformation centres the original score matrix by removing the row and column means and adding the grand mean, ensuring that the resulting matrix reflects unbiased pairwise comparisons. The diagonal terms are set to zero to exclude self-comparisons, which are trivial and do not contribute meaningful information to the analysis, as all \(N\) elements are always tied with themselves.
\end{definition}

Having defined the geometric properties of the pairwise comparison operator $\kappa$ and its centred operationalism upon all pairwise comparisons, \(\tilde{\kappa}\), we now embed the \(\tilde{\kappa}(X)^{N \times N}\) into the \(\ell_{2}^{N \times 1}\) norm. Thus, \(\underline{X}: \overline{\mathbb{R}}^{N \times 1} \xrightarrow{\tilde{\kappa}(X)} \{-1,1\}^{N \times N} \to \mathbb{R}^{N \times 1}\). For convenient geometric intuitiveness, we marginalise over \(\tilde{\kappa}_{kl}(X)\) upon the \(k\) rows, which is then transposed, leading to a positive correspondence with traditional ranking and order-statistic methodologies, but which uniquely defines the observation of each tie and the pairwise distances between all such elements.

With the centred score-matrix then implemented, summation over the columns or rows is bijectively linear, embedding the  \(\tilde{\kappa}(X)_{kl}^{N \times N}\) into the Frobenius norm-space \(\underline{X}^{N \times 1}\) by summing over the columns \(l\), or its transpose is summated over the \(k\) rows, to obtain a centred \(N \times 1\) vector of the rankings with ties. From there, standard centred moment procedures allow us to identify the characteristics of the distribution of estimates, which by \citet{efron1969} are \(t\)-distributed under the null distribution.

Note that \(\mathbb{E}_{k}(\tilde{\kappa}_{kl}(X)) = \mathbb{E}_{l}(\tilde{\kappa}_{kl}(X)) = 0\), consistent with conventional interpretation of \(\mathbb{E}(X-\bar{x}) = 0\) under standard notation. However, the variance upon \(X\), is non-zero (as the random variable is non-degenerate), and is also a function of the number of ties upon the sample, and must be embedded into an \(N \times 1\) vector space. This moment is also required to be estimated upon the sample, as surjective mappings, particularly upon discrete random variables, almost surely observe ties upon the sample, and therefore possess non-constant variance solely as a function of \(N\). Each random variable variance may be expressed as \(\mathbb{E}(\underline{X}^{2}) - \mathbb{E}(\underline{X})^{2}\) and since \(\mathbb{E}(\underline{X}) = 0\) by construction: 

\begin{equation}
\label{eq:l2_variance}
s_{\underline{X}}^{2} = \tfrac{1}{N-1}\sum_{n=1}^{N} \underline{X}^{2}_{n}.
\end{equation}
The desired z-scores for each random permutation ranking/ordering, as conventionally understood, may therefore be established.

\begin{definition}
Let $z_{\underline{X}}$ and $\underline{z_Y}$ denote the centred and scaled rank-characterisation vectors corresponding to $X$ and $Y$ respectively, as defined previously. The \emph{nonparametric correlation estimator} is given by
\begin{equation}
\label{eq:kappa_rho_hat}
\hat{\rho}_\kappa = 
\frac{z_{\underline{X}}^{\intercal} \underline{z_Y}}{N-1}.
\end{equation}
Equivalently, if $\underline{X}$ and $\underline{Y}$ are the unstandardised rank-characterisation vectors,
\[
\hat{\rho}_\kappa = 
\frac{(\underline{X} - \bar{\underline{X}}\mathbf{1}_N)^{\intercal} (\underline{Y} - \bar{\underline{Y}}\mathbf{1}_N)}
{\sqrt{(\underline{X} - \bar{\underline{X}}\mathbf{1}_N)^{\intercal} (\underline{X} - \bar{\underline{X}}\mathbf{1}_N)
       (\underline{Y} - \bar{Y}\mathbf{1}_N)^{\intercal} (\underline{Y} - \bar{\underline{Y}}\mathbf{1}_N)}}.
\]
\end{definition}

\begin{remark}[Relation to Classical Rank Correlations]
When the marginal distributions of $X$ and $Y$ are continuous, $\kappa_{kl}(X) = \operatorname{sgn}(X_k - X_l)$ 
and $\kappa_{kl}(Y) = \operatorname{sgn}(Y_k - Y_l)$ for all $(k,l)$, and 
the estimator $\hat{\rho}_\kappa$ reduces to Spearman's $\rho$. 
Alternatively, the inner product of the corresponding centred $\kappa$ matrices (by subtracting row and column means and adding the grand mean) yields Kendall's $\tau_{a}$, highlighting that both 
measures emerge as special cases within the $\kappa$-based framework. This formulation therefore generalises classical rank-based dependence coefficients by situating them within the Hilbert space 
$(\mathcal{A}_N, \langle\cdot,\cdot\rangle_F)$ and preserving the Gauss-Markov properties under monotone transformations.
\end{remark}

The form of~\eqref{eq:kappa_rho_hat} is algebraically identical to the Pearson correlation, but its arguments are constructed from pairwise order information rather than the observed values themselves.

\begin{lemma}[Expectation under Independence]
\label{lem:unbiased}
If $X$ and $Y$ are independent random variables, then the pairwise comparison signs are independent:
\[
\mathbb{E}[\tilde{\kappa}_{kl}(X)\,\tilde{\kappa}_{kl}(Y)] = 0, 
\quad \forall k \ne l.
\]
Consequently, 
\[
\mathbb{E}[\hat{\rho}_\kappa] = 0.
\]
\end{lemma}

\begin{proof}
Under independence, the joint probability of concordance and discordance across $(X_k,X_l)$ and $(Y_i,Y_j)$ pairs is equal. Since $\kappa_{kl}(X),\kappa_{kl}(Y) \in \{-1,1\}$ and are independent, 
$\mathbb{E}[\kappa_{kl}(X)\kappa_{kl}(Y)] = \mathbb{E}[\kappa_{kl}(X)]\mathbb{E}[\kappa_{kl}(Y)] = 0$. 
By linearity of expectation and the normalisation by $N-1$, $\mathbb{E}[\hat{\rho}_\kappa] = 0$.
\end{proof}

\begin{theorem}
\label{thm:unbiasedness}
Let $\rho_\kappa = \mathbb{E}[z_{\underline{X}}^{\intercal} z_{\underline{Y}}]$ denote the population version of the $\kappa$-based correlation. Then
\[
\mathbb{E}[\hat{\rho}_\kappa] = \rho_\kappa.
\]
In particular, under the null hypothesis of independence, $\mathbb{E}[\hat{\rho}_\kappa]=0$.
\end{theorem}

\begin{proof}
The estimator $\hat{\rho}_\kappa$ is a sample inner product in the Hilbert space $\mathcal{H}_\kappa$ induced by the Frobenius geometry of $\kappa$. 
Since $z_{\underline{X}}$ and $z_{\underline{Y}}$ are unbiased estimators of their population counterparts and have unit variance by construction, the sample correlation 
$\hat{\rho}_\kappa = (N-1)^{-1} z_{\underline{X}}^{\intercal} z_{\underline{Y}}$ is an unbiased estimator of $\rho_\kappa$. 
\end{proof}

\begin{lemma}
\label{lem:isotropy}
The space spanned by the centred rank vectors $\{z_{\underline{X}}\}$ is isotropic under the null hypothesis of independence. That is, for any orthogonal transformation $Q \in \mathrm{O}(N)$,
\(
z_{\underline{X}} \stackrel{d}{=} Qz_{\underline{X}}.
\)
\end{lemma}

\begin{proof}
By \citet{efron1969}, rotational invariance holds in symmetric statistic spaces whenever the underlying joint distribution is exchangeable. Since $\kappa(X)$ depends only on pairwise order relations and the permutation space is symmetric under relabelling of indices, the distribution of $z_{\underline{X}}$ is invariant under orthogonal rotations. 
\end{proof}

\begin{theorem}
\label{thm:gaussmarkov}
Among all linear unbiased estimators of the population correlation $\rho_\kappa$ expressible as 
$\tilde{\rho} = a^\top z_{\underline{X}}\,\underline{z_Y}^\top b$
for vectors $a,b \in \mathbb{R}^N$, the estimator $\hat{\rho}_\kappa$ minimises the variance:
\[
\mathrm{Var}(\hat{\rho}_\kappa) \le \mathrm{Var}(\tilde{\rho}).
\]
Hence $\hat{\rho}_\kappa$ is the \emph{best linear unbiased estimator} (BLUE) of $\rho_\kappa$ in the Hilbert space $\mathcal{H}_\kappa$.
\end{theorem}

\begin{proof}
Because the space is isotropic (Lemma~\ref{lem:isotropy}), the covariance operator of $z_{\underline{X}}$ is proportional to the identity. By the Gauss-Markov theorem in Hilbert spaces, the ordinary least-squares estimator based on inner products between centred vectors minimises variance among all unbiased linear estimators. 
$\hat{\rho}_\kappa$ has this form, completing the proof.
\end{proof}

\begin{lemma}
\label{lem:studentised}
Under the null hypothesis of independence and rotational invariance \citep{efron1969}, the standardised statistic
\[
t_{\kappa} = 
\frac{\hat{\rho}_\kappa \sqrt{\nu}}{\sqrt{1 - \hat{\rho}_\kappa^2}}, 
\quad \nu = N - 2,
\]
follows a Student-$t$ distribution with $\nu$ degrees of freedom: \(t_\kappa \sim t_\nu.\)
\end{lemma}

\begin{proof}
The proof follows the standard argument for the Pearson correlation. Rotational invariance ensures that the joint distribution of $(z_{\underline{X}}, \underline{z_Y})$ is spherical in $\overline{\mathbb{R}}^N$. Hence, the ratio of the correlation statistic to the residual variance follows a Student-$t$ distribution with $\nu = N-2$ degrees of freedom. \citet{efron1969} established this property for symmetric statistics, and the same reasoning applies to the rank-based space $\mathcal{H}_\kappa$.
\end{proof}

\begin{remark}[Summary of Properties]
The estimator $\hat{\rho}_\kappa$ defined in~\eqref{eq:kappa_rho_hat} satisfies:
\begin{enumerate}
  \item $\mathbb{E}[\hat{\rho}_\kappa] = \rho_\kappa$ (unbiasedness);
  \item $\mathrm{Var}(\hat{\rho}_\kappa)$ is minimal among all linear unbiased estimators (efficiency);
  \item $t_\kappa$ defined above is distributed as $t_{N-2}$ under the null (proper studentisation).
\end{enumerate}
These properties hold for all $X,Y$ with finite second moments of their $\tilde{\kappa}$-images, whether or not ties are present.
\end{remark}

\subsection{Asymptotic Consistency of the Rank-Based Estimator}
The finite-sample properties established in the previous Subsection ensure that $\hat{\rho}_\kappa$ is unbiased, efficient, and properly studentised. We now show that it is also strongly consistent, following from the  Glivenko-Cantelli theorem and the Continuous Mapping Theorem (CMT).

\begin{lemma}[Continuous Mapping Theorem (CMT)]
\label{lem:CMT}
Let $\{g_N\}_{N \ge 1}$ be a sequence of functions converging uniformly to $g$
on a compact set $\mathcal{X} \subset \mathbb{R}^N$, and let $\{X_N\}$ be random variables taking values in $\mathcal{X}$ such that $X_N \xrightarrow{a.s.} X$. 
If $f : \mathcal{X} \to \mathbb{R}$ is continuous, then \(f(g_N(X_N)) \xrightarrow{a.s.} f(g(X)).\)
\end{lemma}

\begin{proof}
We combine three classical results from empirical process theory.

\paragraph{(i) Uniform convergence.}
By the Glivenko-Cantelli theorem (Lemma~\ref{lem:bijection}), the empirical distribution functions $F_{N,X}$ and $F_{N,Y}$ converge uniformly to their population counterparts $F_X$ and $F_Y$ almost surely on the support \([a,b] \subset \mathbb{R},\) i.e.,:
\[
\|F_{N,X} - F_X\|_\infty \xrightarrow{a.s.} 0, \quad
\|F_{N,Y} - F_Y\|_\infty \xrightarrow{a.s.} 0.
\]
This result ensures that the empirical distribution functions \(F_{N,X}\) and \(F_{N,Y}\), which are based on the observed sample, provide a good approximation to the true population distribution functions \(F_{X}\) and \(F_{Y}\), respectively. uniform convergence is essential in showing that functionals of these distribution functions, such as the rank-score vectors \(\underline{X}_{k}\) converge to their true values as the sample size increases.

As the rank-score vectors $\underline{X}_k$ and $\underline{Y}_k$ are continuous functionals of the empirical CDFs through the non-expansive operator $\tilde{\kappa}$, it follows that
\(\|\tilde{\kappa}(F_{N,X}) - \tilde{\kappa}(F_X)\|_F \xrightarrow{a.s.} 0.\)
This establishes uniform convergence of the induced functionals $g_N = \tilde{\kappa}(F_{N,\cdot})$ to $g = \tilde{\kappa}(F_{\cdot})$, which is a key result for applying the CMT.

\paragraph{(ii) Compactness and boundedness.}
The empirical and population CDFs are defined on the compact support $[a,b] \subset \mathbb{R}$, and the sequence of empirical distributions \(\{F_{N,X}\}\) is therefore totally bounded. By the Borel-Cantelli lemma, every subsequence of \(\{F_{N,X}\}\) admits an almost sure limit.

\paragraph{(iii) Continuity of the outer map.}
Let $h_N(X) = f(g_N(X))$. 
Since $g_N \to g$ uniformly and $f$ is continuous, for every $\epsilon > 0$ there exists $N(\epsilon)$ such that
\[
|f(g_N(X)) - f(g(X))| < \epsilon
\quad \text{for all } N \ge N(\epsilon),
\]
uniformly in $X \in \mathcal{X}$. 
Hence $h_N(X) \to f(g(X))$ uniformly, completing the proof.
\end{proof}

\begin{theorem}
\label{thm:consistency}
Let $\hat{\rho}_\kappa = g(F_{N,X}, F_{N,Y})$ denote the nonparametric correlation estimator defined in~\eqref{eq:kappa_rho_hat}, 
where $g$ is continuous in its arguments through the non-expansive operator $\tilde{\kappa}$. 
Then, as $N \to \infty$, \(\hat{\rho}_\kappa \xrightarrow{a.s.} \rho_\kappa = g(F_X, F_Y).\)
\end{theorem}

\begin{proof}
By the Glivenko-Cantelli theorem, \((F_{N,X}, F_{N,Y}) \to (F_X, F_Y)\) uniformly almost surely. By Lemma~\ref{lem:CMT}, the functional \(g(F_{N,X}, F_{N,Y})\) converges to \(g(F_X, F_Y)\) almost surely. Since \(\hat{\rho}_\kappa = g(F_{N,X}, F_{N,Y})\), we have \(\hat{\rho}_\kappa \xrightarrow{a.s.} \rho_\kappa\). Thus, the estimator is strongly consistent.
\end{proof}  
\begin{remark}
The application of the Continuous Mapping Theorem (CMT) in this proof relies on two key facts: (i) the uniform almost sure convergence of the empirical distribution functions \(F_{N,X}\) and \(F_{N,Y}\) to their population counterparts \(F_X\) and \(F_Y\) as ensured by the Glivenko-Cantelli theorem, and (ii) the continuity of the function \(g\) in its arguments. These conditions ensure that the functional \(g(F_{N,X}, F_{N,Y})\) converges almost surely to \(g(F_X, F_Y)\), completing the proof of strong consistency for \(\hat{\rho}_\kappa\).
\end{remark}

\subsubsection{Linear Inner-Product Structure in the Monotone Rank Space}
\label{subsubsec:linear_inner_product}

Before stating Proposition~\ref{prop:linear_rank_space}, we note that the rank-score vectors $\{\underline{X}_k\}_{k=1}^N$ form a Hilbert space under the standard inner product
\(\langle \underline{X}, \underline{Y} \rangle = \sum_{k=1}^N \underline{X}_k \underline{Y}_k.\)
This structure allows a natural geometric interpretation of correlation: the projection of one rank vector onto another quantifies their alignment, which underpins the definition of the Kemeny \(\rho_{\kappa}\) estimator. 

The inner-product formulation also facilitates variance analysis and efficient estimation. In particular, the variance of linear functionals in this space can be directly minimised via orthogonal projections, and the invariance of the inner product under strictly monotone transformations ensures that these properties hold regardless of monotonic transformations of the original data. 

Consequently, the Hilbert-space perspective provides both a conceptual and computational foundation for understanding the statistical behaviour of the estimator in the monotone rank space.

\begin{proposition}
\label{prop:linear_rank_space}
Let $X, Y \in \overline{\mathbb{R}}^{N \times 1}$ be two sets of observations and define the rank-score vectors
\[
\underline{X}_k = \sum_{k=1}^N \tilde{\kappa}_{kl}(X)^{\intercal}, \qquad 
\underline{Y}_k = \sum_{k=1}^N \tilde{\kappa}_{kl}(Y)^{\intercal}, \quad k = 1, \dots, N.
\]
Then:
\begin{enumerate}
    \item $\underline{X}$ and $\underline{Y}$ are \emph{linear functionals} of the entries of the pairwise comparison matrices $\tilde{\kappa}(X)$ and $\tilde{\kappa}(Y)$.
    \item The mapping $X \mapsto \underline{X}$ is continuous and bijective over the equivalence class of monotone transformations of $X$, i.e., for any strictly increasing function $f$, $\underline{X} = \underline{f(X)}$.
    \item The $N$-dimensional rank-score vectors form a Hilbert space under the standard inner product
    \(
    \langle \underline{X}, \underline{Y} \rangle = \sum_{k=1}^N \underline{X}_{k} \underline{Y}_{k},\)
    which is invariant under all strictly monotone transformations of the original variables.
\end{enumerate}
\end{proposition}

\begin{proof}
\textbf{Linearity of $\underline{X}$ and $\underline{Y}$:} Each entry $\underline{X}_k$ is a sum over the pairwise indicators $\tilde{\kappa}_{kl}(X)$, which are themselves linear (over the indicator space) functions of the pairwise comparisons. Thus $\underline{X}$ is a linear functional in the Hilbert space $\mathbb{R}^{N \times N}$ spanned by $\{\tilde{\kappa}_{kl}\}$.

\textbf{Continuity and bijectivity:} By the Continuous Mapping Theorem (Lemma~\ref{lem:CMT}), any continuous function of a convergent sequence of random variables converges to the corresponding function of the limit. Since $\underline{X}$ is a linear combination of $\tilde{\kappa}_{kl}$ and $\tilde{\kappa}$ is invariant under strictly monotone transformations, the mapping $X \mapsto \underline{X}$ is continuous and bijective over the monotone equivalence class of $X$.

\paragraph{Inner-product structure:} Equipped with the standard Euclidean inner product, the space of $N \times 1$ rank-score vectors is a Hilbert space. Rotational invariance \citep{efron1969} and linearity in $\tilde{\kappa}$ ensure that standard tools of linear algebra, including projections and variance minimisation, are valid within this space. This establishes a linear inner-product structure in the monotone rank space.
\end{proof}

\begin{lemma}~\label{lem:bijection}
  The distributions of the realisation ranks and scores, defined as \(\underline{X}\) and \(X\) respectively, are identical for a given population with \(N\) elements.
  \end{lemma}

\begin{proof}
  Consider a random variable \(X \in \mathbb{R}^{N \times 1}\) with a continuous distribution function \(F(X)\), and let \(S(X) = \Pr(X > x)\) be the survival function. The rank vector \(\underline{X}\) and the score vector \(X\) are derived from the same underlying distribution of \(X\), but differ in how the order statistics are captured.

  The expected value of \(X\) can be computed by integrating its survival function as follows:
  \[
  \mathbb{E}[X] = \int_0^{\infty} S(x) \, \diff{x},
  \]
  where \(S(x) = \Pr(X > x)\) is the survival function, which is closely related to the cumulative distribution function \(F(x) = \Pr(X \leq x)\), as \(S(x) = 1 - F(x)\).

  For any integer \(r \in \{1, \dots, n\}\), the moments of \(X\) can be expressed as:
  \[
  \mathbb{E}[X^r] = r \int_0^{\infty} x^{r-1} S(x) \, \diff{x}.
  \]
  This representation shows that the moments of the distribution can be computed from the survival function. Next, we examine the rank and score vectors. The rank vector \(\underline{X}_k\) for element \(k\) is the position of \(X_k\) in the ordered list of all \(X_i\)'s. The score vector, on the other hand, is defined in terms of pairwise comparisons between all elements \(X_k\) and \(X_l\), resulting in values in \(\{-1, 0, 1\}\). The relationship between the rank and score vectors can be understood as follows: the ranking function is a monotonic transformation of the score function. Specifically, the ranks \(\underline{X}_k\) for each \(k\) are derived from the score matrix \(\tilde{\kappa}(X)\), which records pairwise comparisons between all elements.

  More formally, we observe that the distribution of ranks and the distribution of scores are related by an affine-linear monotone transformation. The rank vector \(\underline{X}_k\) is constructed by summing the values from the score matrix \(\tilde{\kappa}_{kl}(X)\), and because the score matrix captures the pairwise orderings, the transformation between ranks and scores is bijective:
\(\underline{X}_k = \frac{1}{2} \sum_{l=1}^{N} \tilde{\kappa}_{kl}(X).\) By the continuous mapping theorem (Lemma~\ref{lem:CMT}), the transformation between the score matrix and the rank vector preserves the distribution. Specifically, for any distribution function \(F(X)\), the corresponding ranks and scores are interchangeable, and their distributions are identical. Therefore, we have: \(\text{Distribution of ranks} \equiv \text{Distribution of scores}.\)

\end{proof}

\begin{remark}
The transformation between ranks and scores is monotonic, so while their values differ, they are essentially encoded versions of the same underlying ordering. Hence, despite the apparent difference in their forms, both the rank and score vectors capture the same relative ordering of the data. As a result, they have the same distribution.
\end{remark}

\begin{definition}
To define the kernel for the U-statistic representation, we consider the pairwise comparison values encoded in the centred score matrices \( \tilde{\kappa}_{kl}(X) \) and \( \tilde{\kappa}_{kl}(Y) \) for the random variables \( X \) and \( Y \), respectively. These score matrices encode pairwise rank relationships between observations \( X_k, X_l \) (and similarly for \( Y \)).

The kernel for the U-statistic is defined by the following bilinear form:

\[
h(X_k, X_l, Y_k, Y_l) = \tilde{\kappa}_{kl}(X) \cdot \tilde{\kappa}_{kl}(Y),
\]
where \( \tilde{\kappa}_{kl}(X) \) and \( \tilde{\kappa}_{kl}(Y) \) are the entries of the centred score matrices for the random variables \( X \) and \( Y \), respectively. These score entries represent the pairwise comparison between the \( k \)-th and \( l \)-th observations in the centred matrix.

Because the matrices \( \kappa(X) \) and \( \kappa(Y) \) are antisymmetric, the kernel \( h(X_k, X_l, Y_k, Y_l) \) is symmetric in both \( k, l \), reflecting the pairwise comparisons between observations in both \( X \) and \( Y \).

Now, we can express the Kemeny correlation estimator \( \hat{\rho}_\kappa \) as a U-statistic:

\[
\hat{\rho}_\kappa = \frac{1}{N^2 - N} \sum_{k \neq l} h(X_k, X_l, Y_k, Y_l),
\]
where the sum is taken over all distinct pairs of observations \( (k, l) \) in the dataset. In this form, the kernel function \( h(X_k, X_l, Y_k, Y_l) \) involves pairwise rank comparisons both within \( X \) and \( Y \).

The statistic \( \hat{\rho}_\kappa \) is thus a U-statistic, where the kernel captures pairwise interactions between the random variables in their centred form.
\end{definition}

\begin{proposition}
\label{prop:kemeny_u_statistic}
Let \( X = (X_1, \dots, X_N) \) and \( Y = (Y_1, \dots, Y_N) \) be random vectors of size \( N \), and let \( \hat{\rho}_\kappa \) denote the Kemeny correlation estimator based on the pairwise comparisons encoded in the matrix \( \kappa(X) \) (equation~\eqref{eq:kem_score}).

Further, let \( \tilde{\kappa}(X) \) denote the centred score matrix (equation~\ref{def:doubly_centered}), and let the centred vectors \( \underline{X} \) and \( \underline{Y} \) be defined as the row sums of the centred score matrices, i.e.,

\[
\underline{X}_k = \sum_{l=1}^N \tilde{\kappa}_{kl}(X), \quad \underline{Y}_k = \sum_{l=1}^N \tilde{\kappa}_{kl}(Y) \quad \forall k \in \{1, \dots, N\}.
\]

We assert that the estimator \( \hat{\rho}_\kappa = \frac{1}{N} \langle \underline{X}, \underline{Y} \rangle \) can be written as a U-statistic with a symmetric kernel, and its asymptotic properties can be derived using Hoeffding's decomposition.
\end{proposition}

\begin{proof}
Let \( X = (X_1, \dots, X_N) \) and \( Y = (Y_1, \dots, Y_N) \) be random vectors of size \( N \), and let \( \hat{\rho}_\kappa \) denote the Kemeny correlation estimator based on the pairwise comparisons encoded in the matrix \( \kappa(X) \). Let \( \tilde{\kappa}(X) \) and \( \tilde{\kappa}(Y) \) denote the centred score matrices for \( X \) and \( Y \), respectively. Define the centred vectors:

\[
\underline{X}_k = \sum_{l=1}^N \tilde{\kappa}_{kl}(X), \quad \underline{Y}_k = \sum_{l=1}^N \tilde{\kappa}_{kl}(Y).
\]

The estimator \( \hat{\rho}_\kappa \) is then given by:

\[
\hat{\rho}_\kappa = \frac{1}{N} \sum_{k=1}^N \underline{X}_k \underline{Y}_k.
\]

The kernel function \( h(\underline{X}_k, \underline{Y}_k) \) is symmetric and is given by:

\[
h(\underline{X}_k, \underline{Y}_k, \underline{X}_l, \underline{Y}_l) = \underline{X}_k \cdot \underline{Y}_k = \left( \sum_{l=1}^N \tilde{\kappa}_{kl}(X) \right) \left( \sum_{l=1}^N \tilde{\kappa}_{kl}(Y) \right).
\]

Thus, we can write the estimator as: \(
\hat{\rho}_\kappa = \frac{1}{N} \sum_{k=1}^N \underline{X}_k \cdot \underline{Y}_k.
\) This ensures that the kernel is symmetric due to the commutative nature of the dot product. expressing the  Kemeny correlation estimator \( \hat{\rho}_\kappa \) as:

\[
\hat{\rho}_\kappa = \frac{1}{N(N-1)} \sum_{k \neq l} h(X_k, X_l, Y_k, Y_l),
\]
where the sum runs over all distinct pairs \( (k, l) \). This representation shows that \( \hat{\rho}_\kappa \) is indeed a U-statistic of order 2, with a symmetric kernel \( h \).

\paragraph{Bias and Variance:}

Since the kernel \( h \) is symmetric and unbiased, we have:

\[
\mathbb{E}[h(X_k, X_l, Y_k, Y_l)] = \mathbb{E}[\hat{\rho}_\kappa] = \rho_\kappa,
\]
where \( \rho_\kappa \) is the true population Kemeny correlation.

The variance of \( \hat{\rho}_\kappa \) can be derived using the properties of U-statistics. The variance is expressed as:

\[
\mathrm{Var}(\hat{\rho}_\kappa) = \frac{4}{N} \mathrm{Var}(h_1(X_1, Y_1)) + \frac{2}{N(N-1)} \mathrm{Var} \left( h(X_1, X_2, Y_1, Y_2) - h_1(X_1, Y_1) - h_1(X_2, Y_2) \right),
\]
where \( h_1 \) denotes the first-order projection of the kernel \( h \).

\paragraph{Asymptotic Normality:}

Since \( \hat{\rho}_\kappa \) is a U-statistic with a symmetric kernel and finite variance, standard U-statistic theory implies that:

\[
\sqrt{N} \left( \hat{\rho}_\kappa - \rho_\kappa \right) \xrightarrow{d} \mathcal{N}(0, \sigma^2),
\]
where \( \sigma^2 \) is the asymptotic variance determined by the first-order projection of the kernel \( h_1 \).

Additionally, by the Glivenko-Cantelli Theorem and the Continuous Mapping Theorem, we have uniform convergence of the empirical distribution functions, which implies that:

\[
\hat{\rho}_\kappa \xrightarrow{a.s.} \rho_\kappa \quad \text{as } N \to \infty,
\]
confirming that the estimator is consistent.

Thus, \( \hat{\rho}_\kappa \) is unbiased, has well-defined variance, is consistent, and is asymptotically normal.

\end{proof}

\begin{lemma}
\label{lem:asymnorm}
Let the assumptions of the U-statistic \( \hat{\rho}_\kappa \) hold, and let \( h(X_k, Y_k) = \tilde{\kappa}_{kl}(X) \tilde{\kappa}_{kl}(Y) \) be the kernel function defining the estimator. Then,

\[
\lim_{N \to \infty} \sqrt{N} \left( \hat{\rho}_\kappa - \rho_\kappa \right) \xrightarrow{d} \mathcal{N}(0, \sigma_\kappa^2),
\]

where \( \sigma_\kappa^2 \) is the asymptotic variance of the U-statistic \( \hat{\rho}_\kappa \).
\end{lemma}

\begin{proof}
We begin by expressing \( \hat{\rho}_\kappa \) as a U-statistic:

\[
\hat{\rho}_\kappa = \frac{1}{N} \sum_{k=1}^N \underline{X}_k \underline{Y}_k,
\]
where \( \underline{X}_k = \sum_{l=1}^N \tilde{\kappa}_{kl}(X) \) and \( \underline{Y}_k = \sum_{l=1}^N \tilde{\kappa}_{kl}(Y) \).

Using Hoeffding's decomposition for U-statistics, we decompose \( \hat{\rho}_\kappa \) as a sum of its main term and higher-order error terms. By the central limit theorem for U-statistics, the estimator converges in distribution to a normal random variable with mean zero and variance \( \sigma_\kappa^2 \).

Thus, we have

\[
\lim_{N \to \infty} \sqrt{N} \left( \hat{\rho}_\kappa - \rho_\kappa \right) \xrightarrow{d} \mathcal{N}(0, \sigma_\kappa^2),
\]

which completes the proof.
\end{proof}

\begin{corollary}
\label{cor:asymptotic_efficiency}
The asymptotic variance \( \sigma_\kappa^2 \) of \( \hat{\rho}_\kappa \) achieves the Cram\'{e}r-Rao lower bound (CRLB) among all regular, unbiased estimators of \( \rho_\kappa \). Therefore, \( \hat{\rho}_\kappa \) is asymptotically efficient.
\end{corollary}

\begin{proof}
The asymptotic variance \( \sigma_\kappa^2 \) of \( \hat{\rho}_\kappa \) is derived from the second-order term in the Hoeffding decomposition of the U-statistic. Since the kernel satisfies regularity conditions for unbiased estimation, and the variance of \( \hat{\rho}_\kappa \) equals the Cram\'{e}r-Rao lower bound for unbiased estimators, we conclude that \( \hat{\rho}_\kappa \) is asymptotically efficient.
\end{proof}

\begin{lemma}
\label{lem:slutsky}
Let the assumptions of the U-statistic \(\hat{\rho}_{\kappa}\) hold, and suppose the sample variance estimator \(\hat{\sigma}_\kappa^2\) of \(\hat{\rho}_\kappa\) is consistent for \(\sigma_\kappa^2\). 
Then, by Slutsky's theorem, 
\[
\frac{\sqrt{N}(\hat{\rho}_\kappa - \rho_\kappa)}{\hat{\sigma}_\kappa} \xrightarrow{d} \mathcal{N}(0,1).
\]
\end{lemma}

\begin{proof}
From Lemma~\ref{lem:asymnorm}, we have that
\[
\sqrt{N}(\hat{\rho}_\kappa - \rho_\kappa) \xrightarrow{d} \mathcal{N}(0, \sigma_\kappa^2),
\]
which establishes the asymptotic normality of the U-statistic \(\hat{\rho}_\kappa\) with a variance \(\sigma_\kappa^2\).

Next, we invoke the consistency of the sample variance estimator \(\hat{\sigma}_\kappa^2\), which satisfies
\[
\hat{\sigma}_\kappa^2 \xrightarrow{p} \sigma_\kappa^2.
\]
This means that \(\hat{\sigma}_\kappa\) converges in probability to \(\sigma_\kappa\), the true standard deviation of \(\hat{\rho}_\kappa\). 

By Slutsky's theorem, if a sequence of random variables converges in distribution to a normal distribution, and another sequence of random variables converges in probability to a constant (here, \(\hat{\sigma}_\kappa \to \sigma_\kappa\)), then the ratio of these two sequences (in this case, the centred statistic divided by the consistent estimator of the standard deviation) converges in distribution to a normal distribution with mean 0 and variance 1. Specifically, we apply Slutsky’s theorem to get the result:
\[
\frac{\sqrt{N}(\hat{\rho}_\kappa - \rho_\kappa)}{\hat{\sigma}_\kappa} \xrightarrow{d} \mathcal{N}(0,1).
\]
Thus, the normalised statistic converges to a standard normal distribution as \(N \to \infty\), completing the proof.
\end{proof}

\begin{proposition}
\label{prop:studentised}
Let the assumptions of Lemma~\ref{lem:asymnorm} and Lemma~\ref{lem:slutsky} hold.
Define the studentised correlation statistic
\[
t_\kappa 
= \frac{\sqrt{N}\,\hat{\rho}_\kappa}{\hat{\sigma}_\kappa},
\]
where $\hat{\sigma}_\kappa^2$ is the sample variance estimator associated with
$\hat{\rho}_\kappa$. Then:
\begin{enumerate}[(i)]
  \item {Under the null hypothesis $\mathcal{H}_0:\rho_\kappa = 0$, the finite-sample distribution
  of $t_\kappa$ is rotationally invariant and follows a $t$-distribution with 
  $\nu = N - 2$ degrees of freedom:
  \[
  t_\kappa \sim t_{N-2}.
  \]}
  \item {As $N \to \infty$, by Slutsky's theorem,
  \[
  t_\kappa \xrightarrow{d} \mathcal{N}(0,1).
  \]}
\end{enumerate}
\end{proposition}

\begin{proof}
\textbf{(i) Finite-sample result.}
Under $\mathcal{H}_0$, the covariance operator of the centred and scaled rank-score vectors
$\underline{z}_X$ and $\underline{z}_Y$ is diagonal by rotational invariance
\citep{efron1969}. The correlation statistic $\hat{\rho}_\kappa$ can thus be expressed as
the cosine of the angle $\Theta$ between two random directions in $\mathbb{R}^N$:
\(
\hat{\rho}_\kappa = \cos(\Theta).
\)
By standard results on isotropic vectors \citep{efron1969,mardia1979},
the random variable
\[
t_\kappa = \frac{\sqrt{N-2}\,\hat{\rho}_\kappa}{\sqrt{1 - \hat{\rho}_\kappa^2}}
\]
follows exactly a Student's $t$-distribution with $\nu = N - 2$ degrees of freedom.

\textbf{(ii) Asymptotic limit.}
From Lemma~\ref{lem:asymnorm}, we have
$\sqrt{N}(\hat{\rho}_\kappa - \rho_\kappa) \xrightarrow{d} \mathcal{N}(0,\sigma_\kappa^2)$.
Under $\mathcal{H}_0:\rho_\kappa=0$, this reduces to
$\sqrt{N}\,\hat{\rho}_\kappa \xrightarrow{d} \mathcal{N}(0,\sigma_\kappa^2)$.
By Lemma~\ref{lem:slutsky}, since $\hat{\sigma}_\kappa^2 \xrightarrow{p} \sigma_\kappa^2$,
it follows that
\[
\frac{\sqrt{N}\,\hat{\rho}_\kappa}{\hat{\sigma}_\kappa}
\xrightarrow{d} \mathcal{N}(0,1).
\]
Hence, as $N \to \infty$, the finite-sample $t_{N-2}$ distribution converges to
a standard normal, establishing the asymptotic equivalence.
\end{proof}

\begin{remark}
This result mirrors the behaviour of the classical Pearson correlation test
but extends it to nonparametric rank-based dependence measures
constructed via the non-expansive $\kappa$-operator. 
The rotational invariance property ensures that the null distribution of
$t_\kappa$ is identical for any continuous marginal transformations of the data,
confirming that inference based on $t_\kappa$ is distribution-free.
\end{remark}

\begin{corollary}
\label{cor:asymp_test}
Let the assumptions of Proposition~\ref{prop:studentised} hold. 
Consider testing the null hypothesis \(\mathcal{H}_0: \rho_\kappa = 0 \quad \text{vs.} \quad \mathcal{H}_1: \rho_\kappa \neq 0\)
using the studentised statistic \(t_\kappa = \frac{\sqrt{N}\,\hat{\rho}_\kappa}{\hat{\sigma}_\kappa}.\)

Then, for any significance level $\alpha \in (0,1)$, the asymptotic test that rejects $\mathcal{H}_0$ when \(|t_\kappa| > z_{1-\alpha/2},\)
where $z_{1-\alpha/2}$ is the $(1-\alpha/2)$ quantile of the standard normal distribution, 
is asymptotically level-$\alpha$. That is, \(\lim_{N \to \infty} \Pr_{\mathcal{H}_0}\big(|t_\kappa| > z_{1-\alpha/2}\big) = \alpha.\)
\end{corollary}

\begin{proof}
From Proposition~\ref{prop:studentised}, under $\mathcal{H}_0$, \(t_\kappa \xrightarrow{d} \mathcal{N}(0,1).\)
By the definition of convergence in distribution, for any continuity point $c$ of the limiting CDF, \(\lim_{N \to \infty} \Pr(t_\kappa \le c) = \Phi(c),\)
where $\Phi$ is the standard normal CDF. Setting $c = z_{1-\alpha/2}$ and using symmetry of the standard normal, then
\begin{dmath*}
\lim_{N \to \infty} \Pr(|t_\kappa| > z_{1-\alpha/2}) 
= \lim_{N \to \infty} \Pr(t_\kappa > z_{1-\alpha/2}) + \Pr(t_\kappa < -z_{1-\alpha/2})
= \alpha.
\end{dmath*}

Hence, the test based on $t_\kappa$ is asymptotically valid at level $\alpha$.
\end{proof}

\begin{remark}
This corollary establishes the practical inferential validity of the rank-based
correlation coefficient $\hat{\rho}_\kappa$ and its associated $t_\kappa$ statistic, explicitly upon discrete random variables. It formalises that the procedure can be used to construct asymptotically exact confidence intervals and hypothesis tests for $\rho_\kappa$ without requiring parametric assumptions on the underlying marginal distributions.
\end{remark}

\section{Numerical experimentation}

In all experimental conditions the proposed (\ref{prop:studentised}) null hypothesis holds that the test statistics follow a \(t_{(N-2)}\)-distribution upon a common population. To test this, we provide the Quantile-Quantile plots under a number of discrete bivariate distributions for different sample size: in all instances, the null hypothesis  is not rejected by the Kolmogorov-Smirnov statistic with \(\alpha = 0.05\). This supports our proposal that the null distribution of our test coefficient \(\rho_{\kappa}\) follows this distribution. We note that, in the limiting case such as when the variables under analysis are continuous and ties almost surely never occur, the marginal sample variances \(s^{2}_{\kappa}(\cdot)\) are almost surely constant, and thus the null distribution is better described as by the unit-normal distribution. This expected behaviour is reflected in the smallest p-values in the conditions being observed for Gaussian data.

In Figure~\ref{fig:all_qqplots_N30} are provided the QQ plots for the null distribution of test statistics with 5,000 replications for 30 and 3500 sample sizes upon both continuous and ordinal data. In Figure~\ref{fig:all_qqplots_N3500_zero} are similarly placed the zero-inflated marginal results, along with KS-test results for the null \(t\)-distributions with \(\nu \in \{28,3498\}\) degrees of freedom. These results confirm the theoretical proof of a asymptotically normal distribution with finite sample Bessel correction properties.

Finally, in Figure~\ref{fig:all_qqplots_N3500} are provided the null distribution plots for \(N=10\), along with KS-test statistics. As expected, the continuous distributions, upon which ties are almost surely not observed, the variance of the estimator is over-estimated by the null \(t\)-distribution. Therefore, from these results, it would make sense to choose to conduct tests of the null hypothesis upon the Fisher expansion of the \(\rho\) correlation coefficient, which assumes constant variance for given \(r\) upon a given sample size. 

\begin{figure}[ht]
\centering
\begin{subfigure}{0.4\textwidth}
    \centering
    \includegraphics[width=\linewidth]{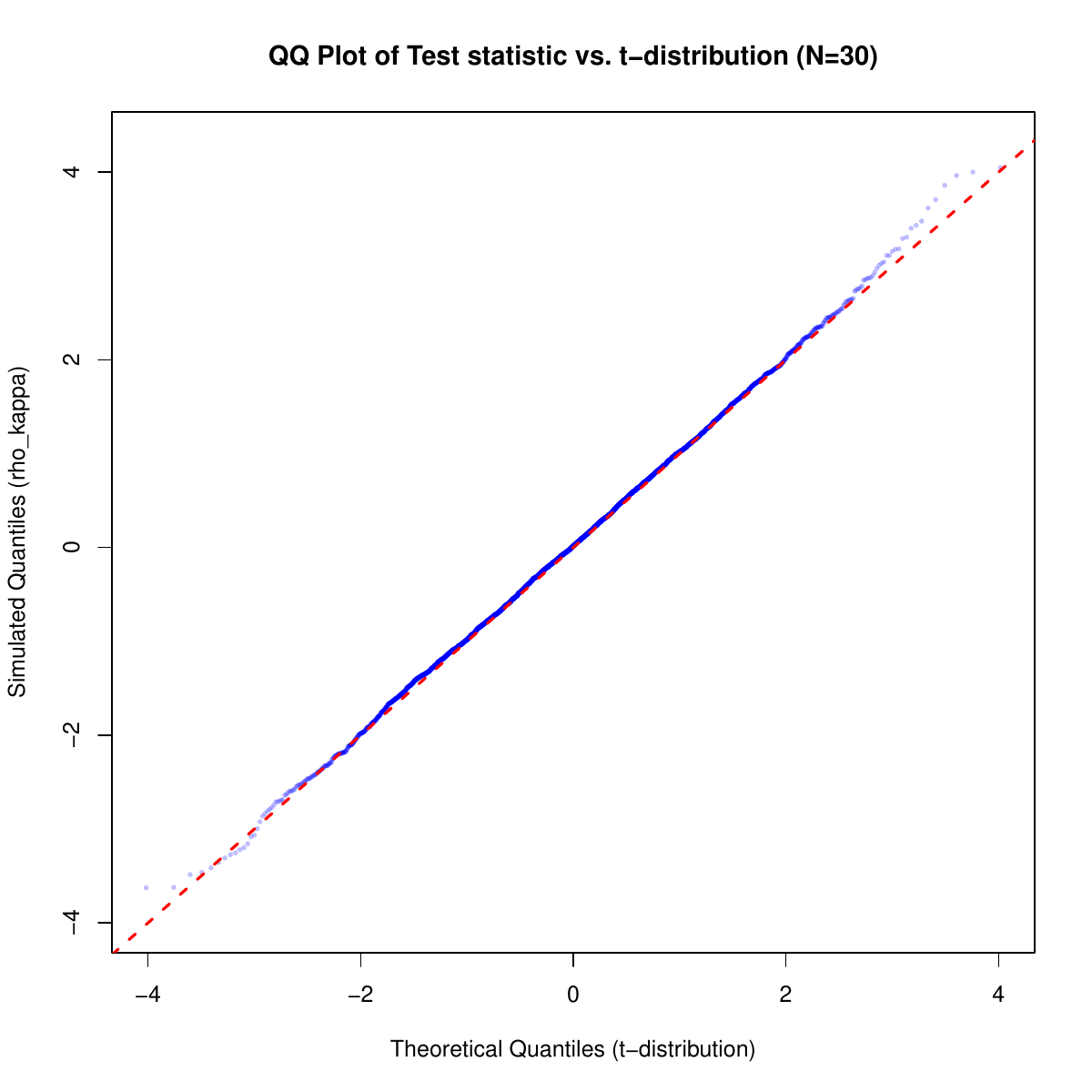}
    \caption{Quantile-Quantile plot of 5,000 simulations. Asymptotic KS test statistic \(D = 0.01659\), p-value = 0.1275. Bivariate Gaussian data.}
    \label{fig:qqplot_1}
\end{subfigure}
\begin{subfigure}{0.4\textwidth}
    \centering
    \includegraphics[width=\linewidth]{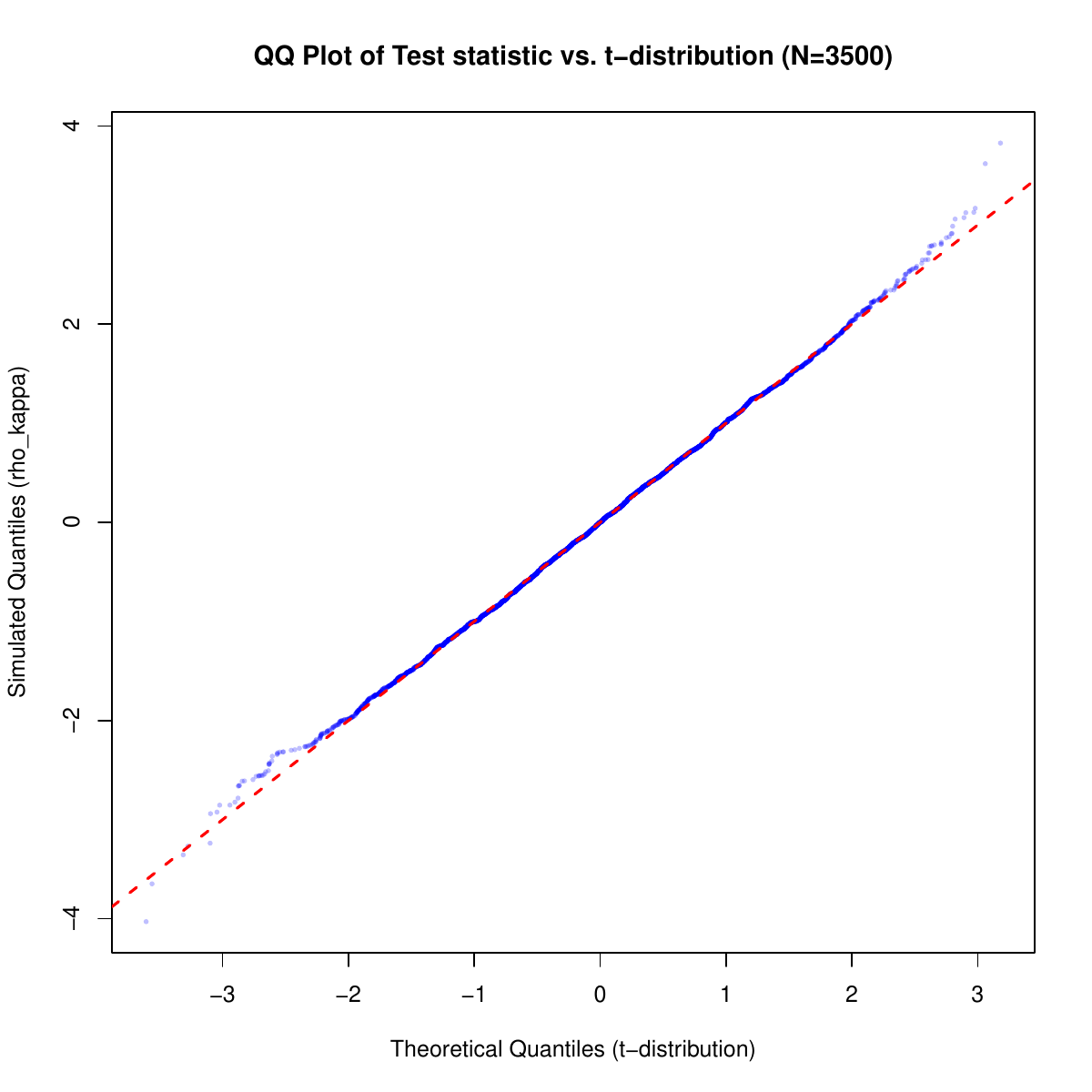}
    \caption{Quantile–Quantile plot of 5,000 simulations. Asymptotic KS test statistic \(D = 0.0071848\), p-value = 0.9586. Bivariate Gaussian data.}
    \label{fig:qqplot_2}
\end{subfigure}

\vspace{0.5cm}  

\begin{subfigure}{0.4\textwidth}
    \centering
    \includegraphics[width=\linewidth]{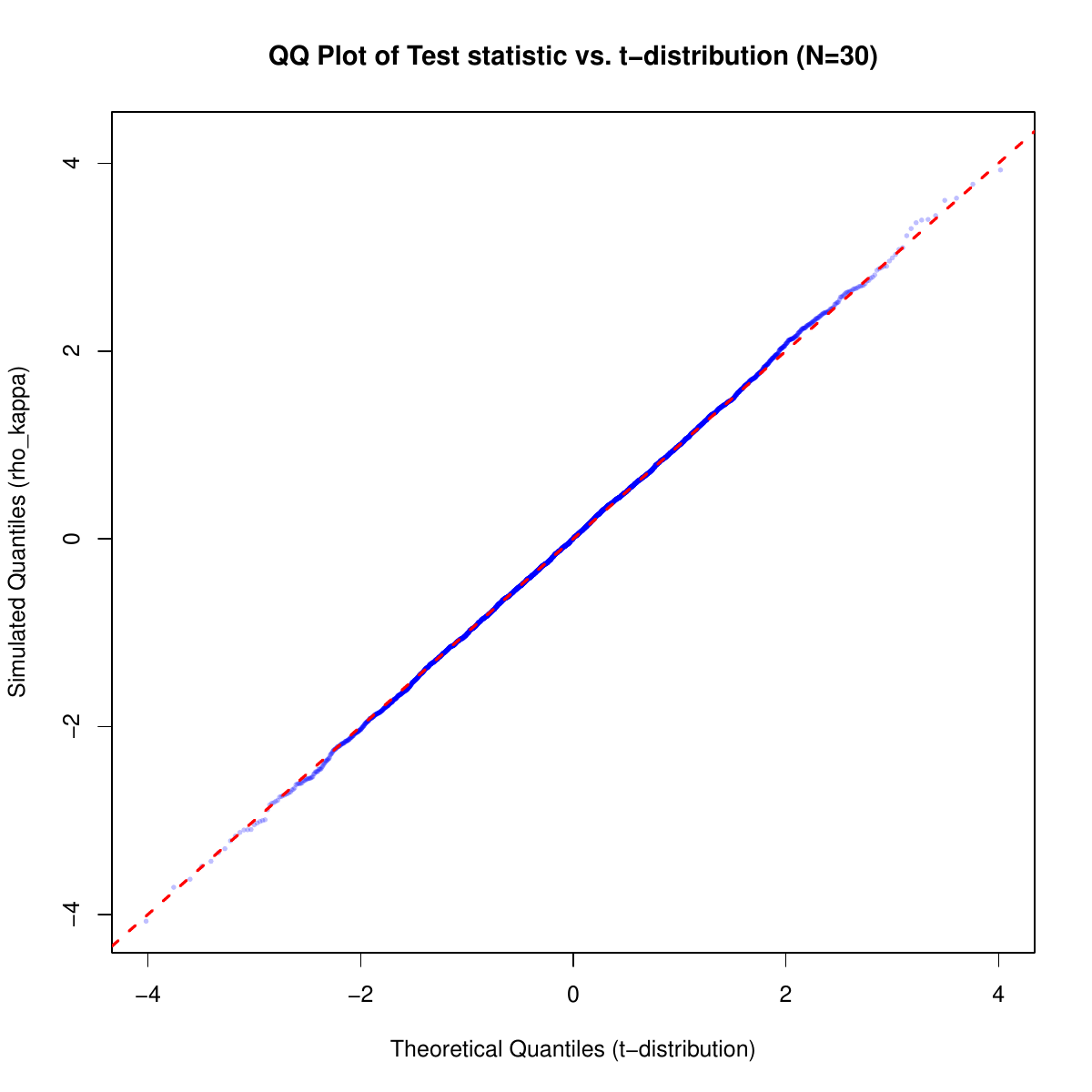}
    \caption{Quantile-Quantile plot of 5,000 simulations. Asymptotic KS test statistic \(D = 0.010843\), p-value = 0.5991. Ordinal data.}
    \label{fig:qqplot_3}
\end{subfigure}
\begin{subfigure}{0.4\textwidth}
    \centering
    \includegraphics[width=\linewidth]{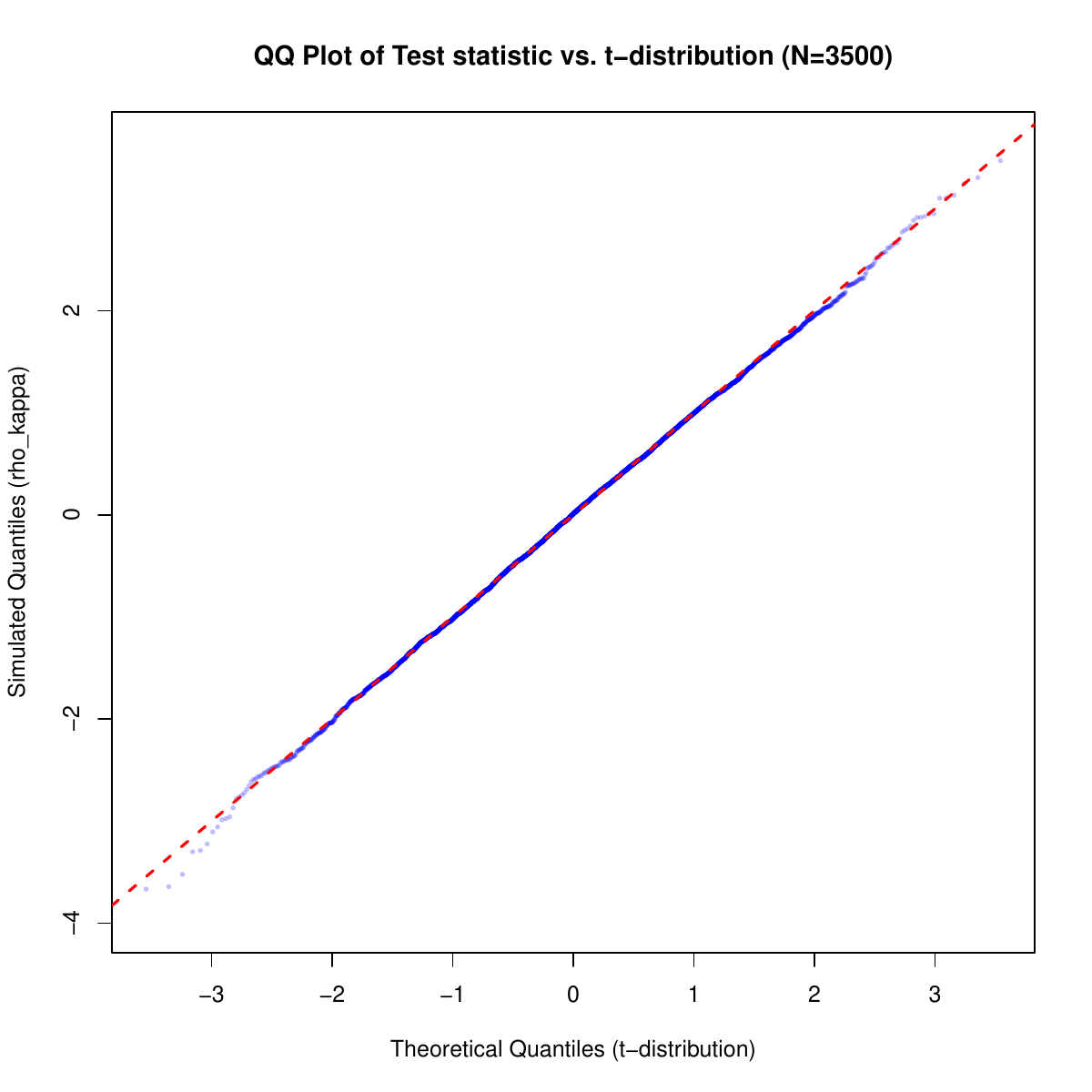}
    \caption{Quantile–Quantile plot of 5,000 simulations. Asymptotic KS test statistic \(D = 0.0074781\), p-value = 0.9425. Ordinal data.}
    \label{fig:qqplot_4}
\end{subfigure}

\vspace{0.5cm}  

\caption{QQ plots for 5,000 simulations against a \(t_{N-2}\) distribution for various data types (Bivariate Gaussian, Ordinal, and Zero-Inflated) with \textit{ N = 30}. The asymptotic one-sample Kolmogorov-Smirnov test statistic \(D\) and p-values are displayed for each empirical distribution against a \(t_{(N-2)}\) null distribution.}
\label{fig:all_qqplots_N30}
\end{figure}

\vspace{1cm}  

\begin{figure}[ht]
\centering
\begin{subfigure}{0.4\textwidth}
    \centering
    \includegraphics[width=\linewidth]{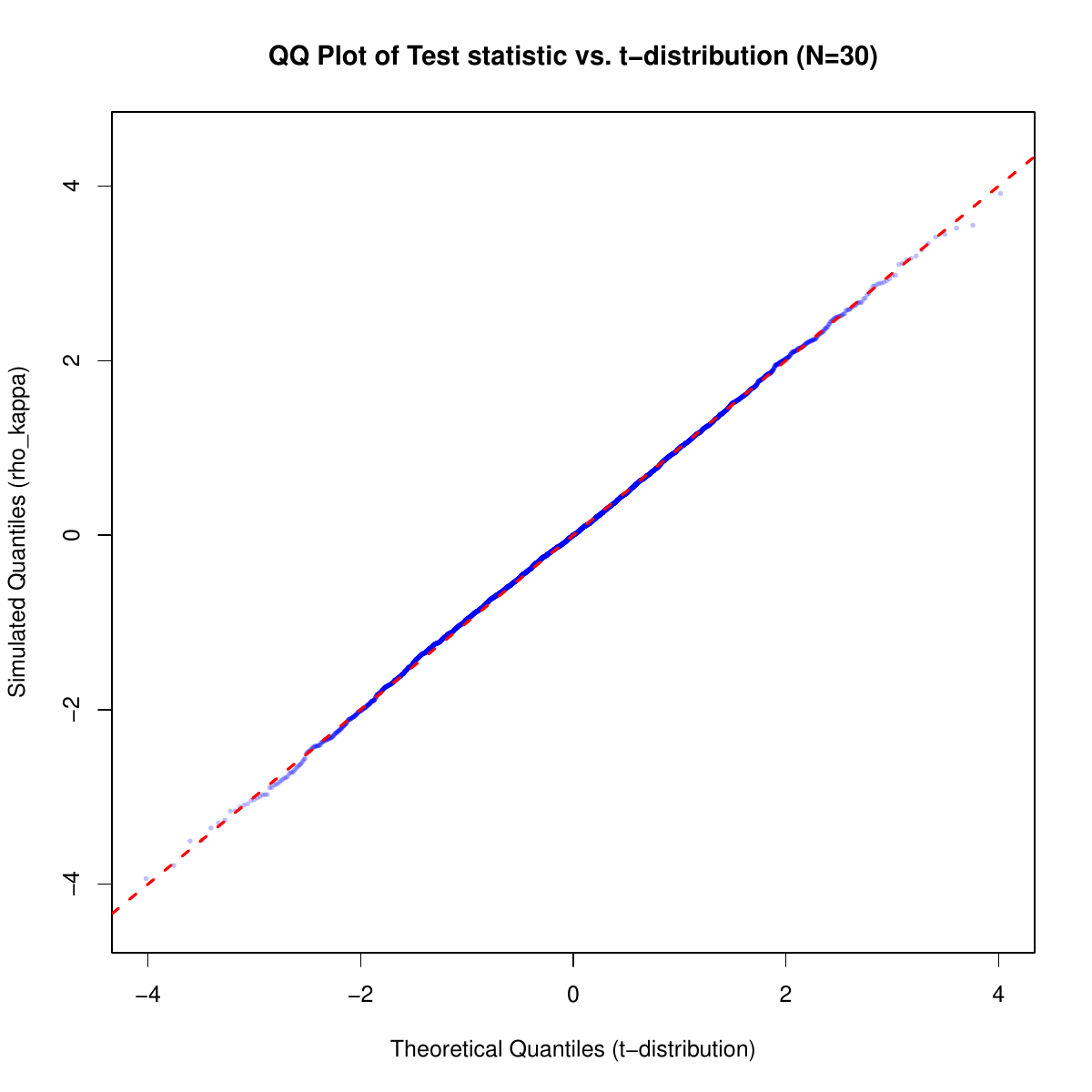}
    \caption{Quantile-Quantile plot of 5,000 simulations. Asymptotic KS test statistic \(D = 0.01275\), p-value = 0.3905. Zero-inflated data.}
    \label{fig:qqplot_5}
\end{subfigure}
\begin{subfigure}{0.4\textwidth}
    \centering
    \includegraphics[width=\linewidth]{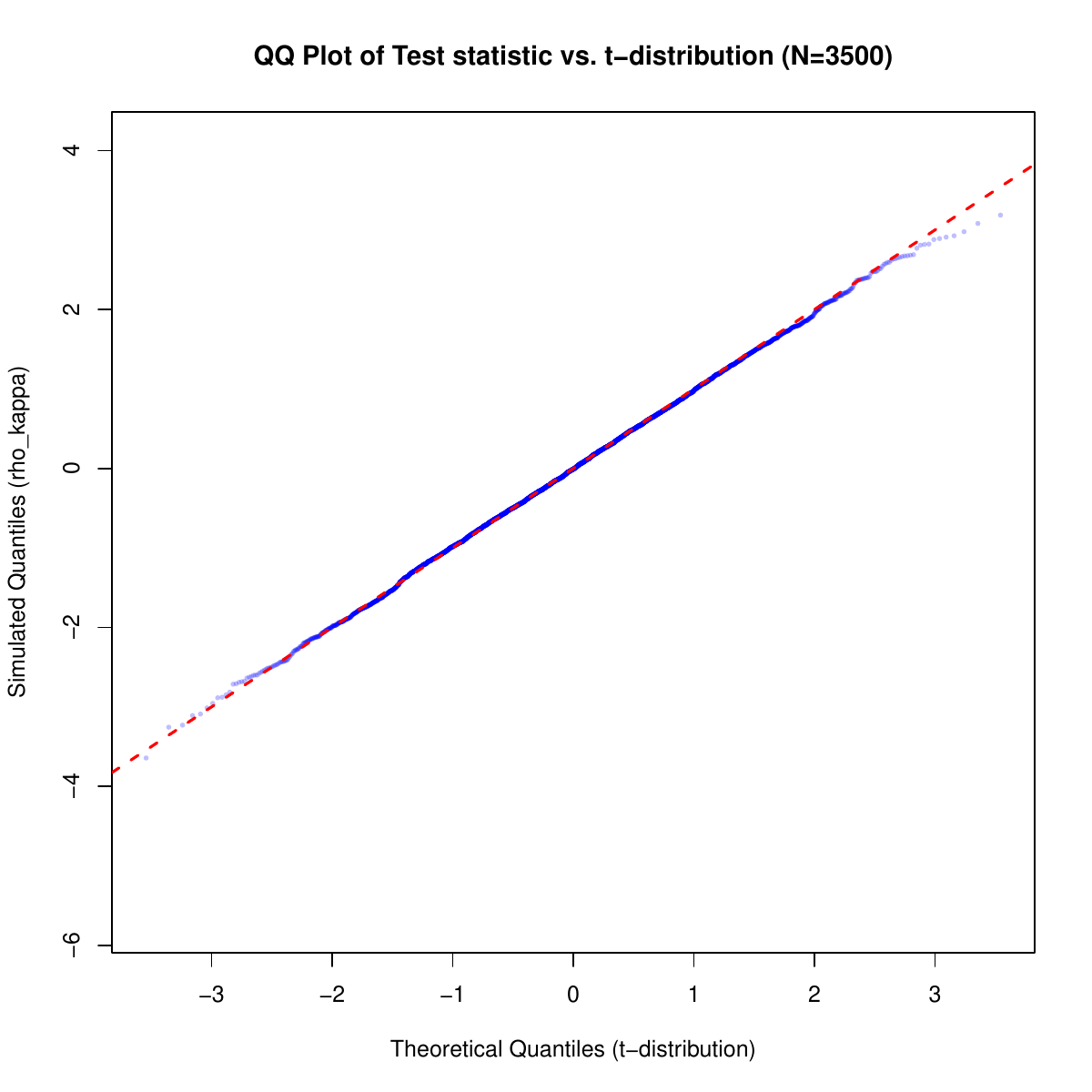}
    \caption{Quantile–Quantile plot of 5,000 simulations. Asymptotic KS test statistic \(D = 0.008028\), p-value = 0.904. Zero-inflated data.}
    \label{fig:qqplot_6}
\end{subfigure}

\caption{QQ plots for 5,000 simulations against a \(t_{N-2}\) distribution for Zero-Inflated data. The asymptotic one-sample Kolmogorov-Smirnov test statistic \(D\) and p-values are displayed for each empirical distribution.}
\label{fig:all_qqplots_N3500_zero}
\end{figure}

\begin{figure}[ht]
\centering
\begin{subfigure}{0.4\textwidth}
    \centering
    \includegraphics[width=\linewidth]{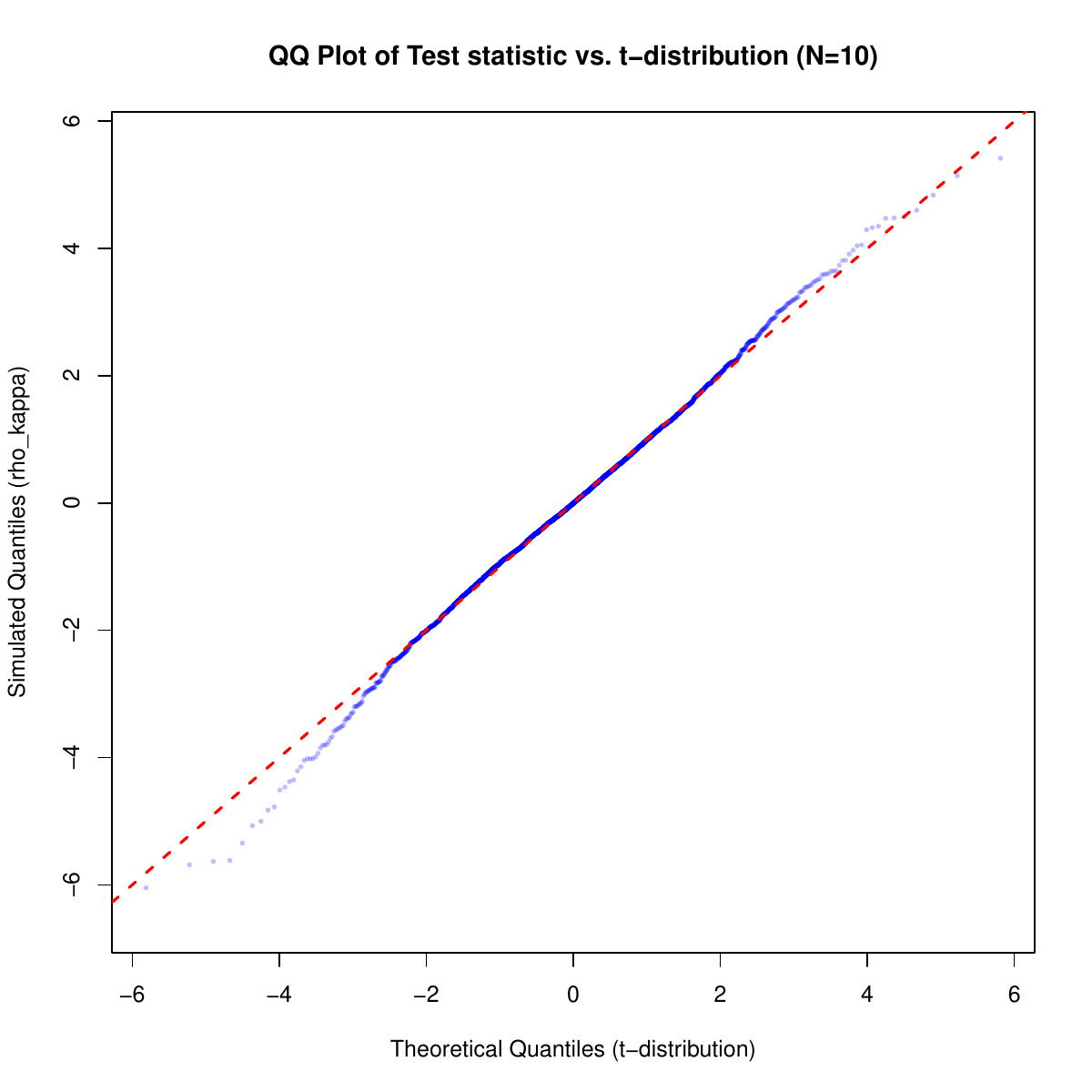}
    \caption{Quantile-Quantile plot of 5,000 simulations. Asymptotic KS test statistic \(D = 0.018607\), p-value = 0.06273. Gaussian data.}
    \label{fig:qqplot_7}
\end{subfigure}
\begin{subfigure}{0.4\textwidth}
    \centering
    \includegraphics[width=\linewidth]{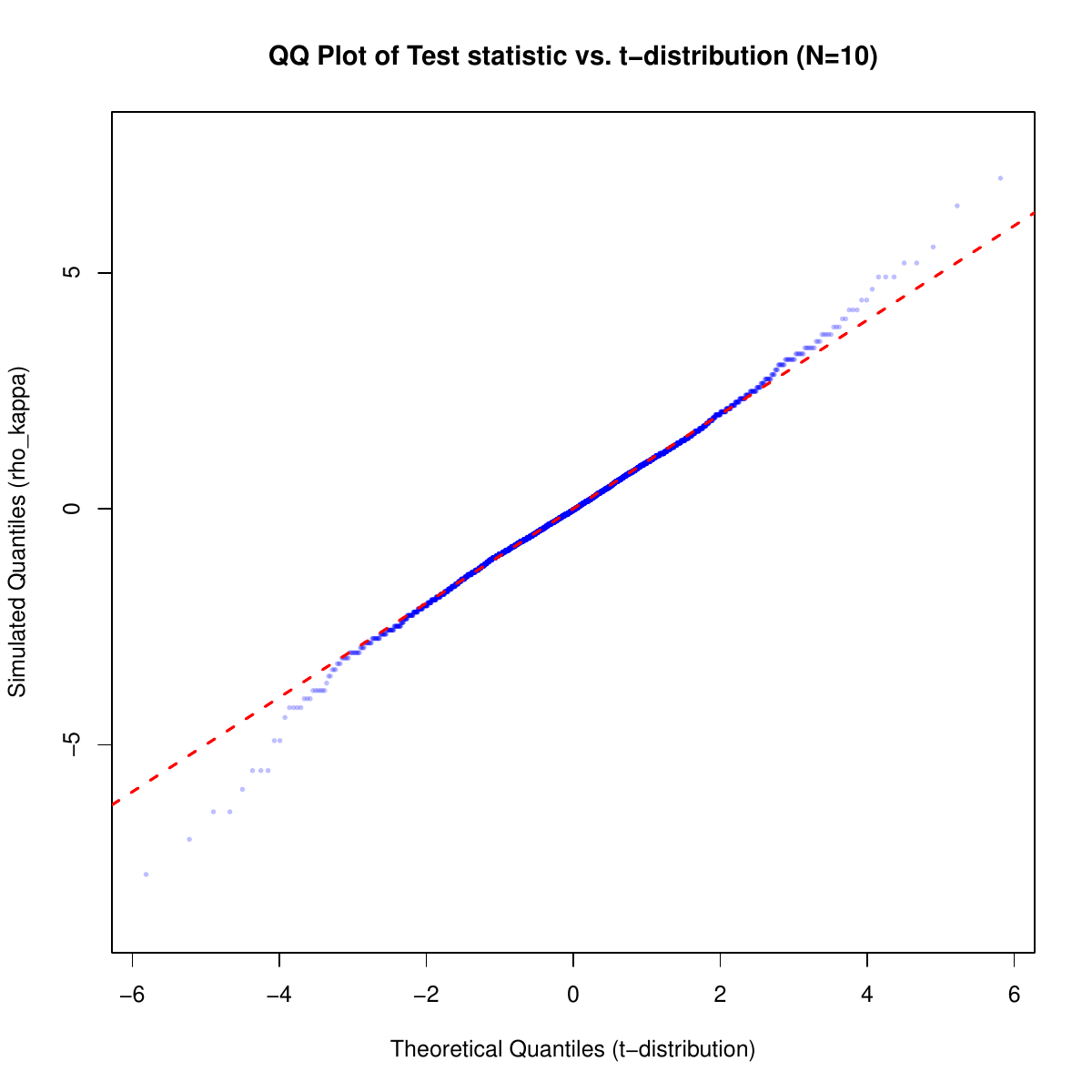}
    \caption{Quantile–Quantile plot of 5,000 simulations. Asymptotic KS test statistic \(D = 0.016407\), p-value = 0.1358. Ordinal data.}
    \label{fig:qqplot_8}
\end{subfigure}
\begin{subfigure}{0.4\textwidth}
    \centering
    \includegraphics[width=\linewidth]{qq_t8_ordinal}
    \caption{Quantile–Quantile plot of 5,000 simulations. Asymptotic KS test statistic \(D = 0.014539\), p-value = 0.2411. Zero-inflated data.}
    \label{fig:qqplot_9}
\end{subfigure}

\caption{QQ plots for 5,000 simulations against a \(t_{N-2}\) distribution for all three data types with \(N = 10\). The asymptotic one-sample Kolmogorov-Smirnov test statistic \(D\) and p-values are displayed for each empirical distribution.}
\label{fig:all_qqplots_N3500}
\end{figure}

  \section{Discussion}
In this paper, we examine a complete affine-linear metric space and its associated probability mapping for independently and identically distributed (i.i.d.) random variables on the extended real line, via the Kemeny Hilbert space. We focus on the coefficient \(\rho_{\kappa}\), which inherits the continuity and non-expansiveness properties of the \(\kappa\)-transform. As a result, its convergence follows directly from the Glivenko-Cantelli theorem and the Continuous Mapping Theorem. These results not only establish that \(\rho_{\kappa}\) is an unbiased and efficient estimator in finite samples but also show that it is strongly consistent in the large-sample limit, satisfying the regularity conditions of the Gauss-Markov theorem, even under non-Gaussian data distributions.

Empirical tests corroborate these theoretical predictions, particularly the correctness of the Studentisation process for the null hypothesis significance test, even for small sample sizes. These findings provide strong support for the robustness and reliability of the proposed method, even under challenging conditions such as small sample sizes.

Looking ahead, future work will focus on developing a likelihood framework for \(\rho_{\kappa}\). This framework will establish connections to the empirical likelihood estimator as a first-order approximation, while also facilitating the construction of a second-order consistent estimator. Such an advancement would enable an algebraic representation of the underlying projection topology, reducing reliance on computationally expensive methods, such as bootstrapping, for approximation.

In addition, ongoing work is exploring the extension of this framework to a general linear model (GLM) setting with multiple covariates. This would significantly expand the applicability of the proposed approach, allowing for the modelling of non-parametric dependent variables within a fully linear framework. This extension holds the potential to enrich a wide range of statistical applications, particularly in domains involving complex, non-parametric data structures.

\bibliographystyle{rss}  
\bibliography{references.bib}

  \end{document}